\documentclass[11pt]{article}
\usepackage{amsmath, amsfonts, url}
\usepackage[T1]{fontenc}
\usepackage[latin1]{inputenc}
\usepackage{amssymb}

\begin{document}

\makeatletter

\newenvironment{algorithm}{\begin{algorithm1}\ \\
    \vspace{-0.2cm}}{\end{algorithm1}}

\newenvironment{proofsk}{\begin{proof}[Proof Sketch:]}
{\end{proof}}

\newenvironment{smallproof}{\nopagebreak \begin{quote} %
\begin{small} \noindent{\bf Proof:}}{ \qed \par %
\end{small} \end{quote} \medskip}

\newenvironment{note}{\nopagebreak \begin{quote} %
\noindent{\bf Note:}}{%
\end{quote} \medskip}

\newenvironment{notes}{\nopagebreak \begin{quote} %
\noindent{\bf Notes:} \par%
\begin{itemize}}{%
\end{itemize}\end{quote} \medskip}

\newenvironment{summary}{\begin{quote} {\bf Summary:}}{\end{quote}}

%%%%%%%%%%%%%%%%%%%%%%%%%%%%%%%%%
% General Macros

\newcommand{\eqdef}{\stackrel{def}{=}}
\newcommand{\N}{\mathbb{N}}
\newcommand{\R}{\mathbb{R}}
\newcommand{\C}{\mathbb{C}}
\newcommand{\Z}{\mathbb{Z}}
\newcommand{\F}{\mathbb{F}}
\newcommand{\Zn}{{\Z}_n}
\newcommand{\bits}{\{0,1\}}
\newcommand{\inr}{\in_{\mbox{\tiny R}}}
\newcommand{\getsr}{\gets_{\mbox{\tiny R}}}
\newcommand{\st}{\mbox{ s.t. }}
\newcommand{\etal}{{\it et al }}
\newcommand{\into}{\rightarrow}

\newcommand{\Ex}{\mathbb{E}}
\newcommand{\To}{\rightarrow}
\newcommand{\e}{\epsilon}
\newcommand{\ee}{\varepsilon}
\newcommand{\ceil}[1]{{\lceil{#1}\rceil}}
\newcommand{\floor}[1]{{\lfloor{#1}\rfloor}}
\newcommand{\angles}[1]{\langle #1 \rangle}
\newcommand{\var}{\mbox{var}}
\newcommand{\trace}{\mbox{trace}}
\newcommand{\ignore}[1]{}
\newcommand{\Alg}{\mathrm{Alg}}
\newcommand{\fro}[1]{\|#1\|_F}
\newcommand{\trn}[1]{\|#1\|_{tr}}
\newcommand{\norm}[1]{\|#1\|}

% complexity classes
\newcommand{\NP}{\mathbf{NP}}
\renewcommand{\P}{\mathbf{P}}
\newcommand{\PCP}{\mathbf{PCP}}
\newcommand{\RP}{\mathbf{RP}}
\newcommand{\BPP}{\mathbf{BPP}}

\newcommand{\Tr}{\mathrm{Tr}}

\newcommand{\lang}{\mathcal{L}}

\newcommand{\Poincare}{Poincar�}

%%%%%%%%%%%%%%%%%%%%%%%%%%%%%%%%%%%%%%%%%%%%%%%%%%%%%%%%%%%%%%%%%%%%%%%%%%%%%%%%%%%%%%%%%%%%%%%%%%%%%%%%%%%%%
%% Environment definitions
%%%%%%%%%%%%%%%%%%%%%%%%%%%%%%%%%%%%%%%%%%%%%%%%%%%%%%%%%%%%%%%%%%%%%%%%%%%%%%%%%%%%%%%%%%%%%%%%%%%%%%%%%%%%%

\newtheorem{theorem}{Theorem}
\newtheorem{lemma}[theorem]{Lemma}
\newtheorem{fact}[theorem]{Fact}
\newtheorem{claim}[theorem]{Claim}
\newtheorem{corollary}[theorem]{Corollary}
\newtheorem{conjecture}[theorem]{Conjecture}
\newtheorem{question}[theorem]{Question}
\newtheorem{proposition}[theorem]{Proposition}
\newtheorem{axiom}[theorem]{Axiom}
\newtheorem{remark}[theorem]{Remark}
\newtheorem{example}[theorem]{Example}
\newtheorem{exercise}[theorem]{Exercise}
\newtheorem{definition}[theorem]{Definition}
\newtheorem{observation}[theorem]{Observation}

\def\pproof{\par\penalty-1000\vskip .5 pt\noindent{\bf Proof\/ }}
\newcommand{\QED}{\hfill$\;\;\;\rule[0.1mm]{2mm}{2mm}$}

\newenvironment{proof}{\begin{pproof}}{\QED\end{pproof}~\\}

%%%%%%%%%%%%%%%%%%%%%%%%%%%%%%%%%%%%%%%%%%%%%%%%%%%%%%%%%%%%%%%%%%%%%%%%%%%%%%%%%%%%%%%%%%%%%%%%%%%%%%%%%%%%%
%% Macro definitions
%%%%%%%%%%%%%%%%%%%%%%%%%%%%%%%%%%%%%%%%%%%%%%%%%%%%%%%%%%%%%%%%%%%%%%%%%%%%%%%%%%%%%%%%%%%%%%%%%%%%%%%%%%%%%

\newcommand{\NChooseM}[2]{\ensuremath{\lp{(}\begin{array}{cc}#1\\#2\end{array}\rp{)}}}
\newcommand{\mc}[1]{{\cal{#1}}}
\newcommand{\lp}[1]{\left #1}
\newcommand{\rp}[1]{\right #1}
\newcommand{\vect}[1]{\ensuremath{{\mathbf #1}}}
\newcommand{\IP}[2]{\ensuremath{\lp{<}#1,#2\rp{>}}}
% sign vector
\newcommand{\sv}{\ensuremath{\vect{v}}}
% sign vector entry
\newcommand{\sve}{\ensuremath{v}}

\newcommand{\op}{W}
\newcommand{\opr}{\vect{w}}
\newcommand{\allvects}[1]{L_{#1}}
\newcommand{\avn}[1]{L_{#1}}
\newcommand{\av}{L_n}
\newcommand{\ff}{\lfloor f \rfloor}

\newcommand{\sizeo}[1][\e]{\mathrm{size}^{(v)}_{#1}}
\newcommand{\corro}[2][\e]{\mathrm{ubc}^{(#2)}_{#1}}
\newcommand{\corr}[1][\e]{\mathrm{ubc}_{#1}}
\newcommand{\rcorro}[2][\e]{\mathrm{corr}^{(#2)}_{#1}}
\newcommand{\rcorr}[1][\e]{\mathrm{corr}_{#1}}
\newcommand{\size}[1][\e]{\mathrm{size}_{#1}}
\newcommand{\scorr}[1][\e]{\mathrm{sc}_{#1}}
\newcommand{\monoo}[2][\rho]{\mathrm{mono}^{(#2)}_{#1}}
\newcommand{\mono}[1][\rho]{\mathrm{mono}_{#1}}
\newcommand{\hmonoo}[2][\rho]{\mathrm{hmono}^{(#2)}_{#1}}
\newcommand{\hmono}[1][\rho]{\mathrm{hmono}_{#1}}

\newcommand{\rk}{\mathrm{rank}}

\newcommand{\ind}{\ensuremath{\alpha}}
\newcommand{\chr}{\ensuremath{\chi}}

\newcommand{\rs}{Ruzsa-Szemer\'{e}di}

\title{A Note on Multiparty Communication Complexity and the Hales-Jewett Theorem}

\author{Adi Shraibman \\ The School of Computer Science \\  
The Academic College of Tel Aviv-Yaffo\\{\tt adish@mta.ac.il}}

\date{}

\maketitle

\abstract{
For integers $n$ and $k$, the {\em density Hales-Jewett number} $c_{n,k}$
is defined as the maximal size of a subset of $[k]^n$ that contains no combinatorial
line. We show that for $k \ge 3$ the density Hales-Jewett number $c_{n,k}$ is equal to the 
maximal size of a cylinder intersection in the problem $Part_{n,k}$ of testing whether $k$ subsets of $[n]$ form a partition.
It follows that the communication complexity, in the 
Number On the Forehead (NOF) model, of $Part_{n,k}$, is equal to the minimal 
size of a partition of $[k]^n$ into subsets that do not contain a combinatorial line.  
Thus, the bound in \cite{chattopadhyay2007languages} on $Part_{n,k}$ using
the Hales-Jewett theorem is in fact tight, and the density Hales-Jewett number
can be thought of as a quantity in communication complexity. This gives a new angle to this well studied
quantity.

As a simple application we prove a lower bound on $c_{n,k}$, similar to the lower bound in \cite{polymath2010moser}
which is roughly $c_{n,k}/k^n \ge \exp(-O(\log n)^{1/\lceil \log_2 k\rceil})$. This lower bound follows from a protocol
for $Part_{n,k}$. It is interesting to better understand the communication complexity of $Part_{n,k}$ as this
will also lead to the better understanding of the Hales-Jewett number.  
The main purpose of this note is to motivate this study.
}

\section{Introduction}

For any integers $n \ge 1$ and $k \ge 1$, consider the set $[k]^n$, of {\em words} 
of length $n$ over the alphabet $[k]$.
Define a {\em combinatorial line} in $[k]^n$ as a subset of $k$ distinct words such that we can place
these words in a $k \times n$ table so that all columns in this table belong to the set
$\{ (x,x,\ldots,x) : x\in [k] \} \cup \{ (1,2,\ldots,k) \}$. The   
{\em density Hales-Jewett number} $c_{n,k}$ is defined to be the maximal cardinality
of a subset of $[k]^n$ which does not contain a combinatorial line.

Clearly, $c_{n,k} \le k^n$, and a deep theorem of Furstenberg and Katznelson 
\cite{furstenberg1989density, furstenberg1991density} says that $c_{n,k}$ 
is asymptotically smaller than $k^n$:
\begin{theorem}[Density Hales-Jewett theorem]
For every positive integer $k$ and every real number $\delta > 0$ there exists
a positive integer $DHJ(k, \delta)$ such that if $n \ge DHJ(k,\delta)$ then any subset
of $[k]^n$ of cardinality at least $\delta k^n$ contains a combinatorial line.
\end{theorem}
The above theorem is a density version of the Hales-Jewett theorem:
\begin{theorem}[Hales-Jewett theorem]
For every pair of positive integers $k$ and $r$ there exists a positive number $HJ(k,r)$
such that for every $n \ge HJ(k,r)$ and every $r$-coloring of the set $[k]^n$ there is a
monochromatic combinatorial line.
\end{theorem}
 
Note that the density Hales-Jewett theorem implies the Hales-Jewett theorem
but not the other way around.  
The density Hales-Jewett theorem is a fundamental result of Ramsey theory.
It implies several well known results, such as van der Waerden's theorem \cite{van1927beweis}, Szemer\'edi's
theorem on arithmetic progressions of arbitrary length \cite{szemeredi1975sets} and its multidimensional
version \cite{furstenberg1978ergodic}.

The proof of Furstenberg and Katznelson used ergodic-theory and gave no explicit bound
on $c_{n,k}$. Recently, additional proofs of this theorem were found \cite{polymath2009new, austin2011deducing, dodos2013simple}. The proof of
\cite{polymath2009new} is the first combinatorial proof of the density Hales-Jewett theorem, 
and also provides effective bounds for $c_{n,k}$. In a second paper
\cite{polymath2010moser} in this project, several values of $c_{n,3}$ are 
computed for small values of $n$.
Using ideas from recent work \cite{elkin2010improved, green2010note, o2011sets} on the
construction of Behrend \cite{behrend1946sets} and Rankin \cite{rankin1965sets},
they also prove the following asymptotic bound on $c_{n,k}$.
Let $r_k(n)$ be the maximal size of a subset of $[n]$ without an arithmetic progression of length $k$,
then: 
\begin{theorem}[\cite{polymath2010moser}]
For each $k \ge 3$, there is an absolute constant $C > 0$ such that
\[
c_{n,k} \ge Ck^n \left( \frac{r_k(\sqrt{n})}{\sqrt{n}}\right)^{k-1} = k^n \exp\left(-O(\log n)^{1/\lceil \log_2 k \rceil}\right).
\]
\end{theorem}

We show analogues of the Hales-Jewett theorem, the density Hales-Jewett theorem, the above
lower bound and other related quantities, in the communication complexity framework.
The model used is the Number On the 
Forehead (NOF) model \cite{CFL83}. In this model
$k$ players compute together a boolean function 
$f:X_1\times\cdots \times X_k \to \{0,1\}$.
The input, $(x_1,x_2,\ldots,x_k) \in X_1\times\cdots \times X_k$, is presented to the players 
in such a way that the $i$-th player sees the entire input except $x_i$. A protocol is comprised of rounds, in each
of which every player writes one bit ($0$ or $1$) on a board that is visible to all players.
The choice of the written bit may
depend on the player's input and on all bits previously written by himself and others on the board. 
The protocol ends when all players know $f(x_1,x_2,\ldots,x_k)$.
The cost of a protocol is the number of bits written on the board, for the worst input. The deterministic
communication complexity of $f$, $D(f)$, is the cost of the best protocol for $f$.

Two key definitions in the number on the forehead model are a {\em cylinder} and 
a {\em cylinder intersection}. We say that $C \subseteq X_1\times\cdots \times X_k$
is a cylinder in the $i$-th coordinate if membership in $C$ does not depend on the $i$-th coordinate.
Namely, for every $y,y'$ and $x_1,x_2,\ldots,x_{i-1},x_{i+1},\ldots,x_k$ there holds 
$(x_1,x_2,\ldots,x_{i-1},y,x_{i+1},\ldots,x_k)\in C$ iff $(x_1,x_2,\ldots,x_{i-1},y',x_{i+1},\ldots,x_k)\in C$.
A cylinder intersection is a set $C$ of the form $C = \cap_{i=1}^k C_i$
where $C_i$ is a cylinder in the $i$-th coordinate.

Every $c$-bit communication protocol for a function $f$ partitions the input space
into at most $2^c$ cylinder intersections that are monochromatic with respect to 
$f$ (see \cite{KN97} for more details). Thus, one way to relax $D(f)$ is to view it as 
a coloring problem. Denote by $\ind(f)$ the largest size of a $1$-monochromatic cylinder intersection
with respect to $f$, and by $\chr(f)$ the least number of monochromatic cylinder intersections
that form a partition of $f^{-1}(1)$. Obviously, $D(f) \ge \log \chr(f)$, and as we shall see,
for special families of functions this bound is nearly tight, including the function $Part_{n, k}$ that we are interested in.
Also observe that $\chr(f) \ge |f^{-1}(1)|/\ind(f)$.

The function $Part_{n,k}: (2^{[n]})^k \to \{0,1\}$ is defined as follows, $Part_{n,k}(S_1,\ldots,S_k) = 1$ if and only if
$(S_1,\ldots,S_k)$ is a partition of $[n]$. In \cite{chattopadhyay2007languages} the Hales-Jewett theorem was used to prove that $D(Part_{n,k}) \ge \omega(1)$.
We observe that in fact the Hales-Jewett theorem is equivalent to this statement. This follows from the following strong
relation $Part_{n,k}$ has with the Hales-Jewett number.

\newpage
\begin{theorem}
\label{th:1}
For every $k \ge 3$ and $n \ge 1$ there holds:

\begin{enumerate}

\item $c_{n,k} = \ind(Part_{n,k})$, and

\item $\chr(Part_{n,k})$ is equal to the minimal number of
colors required to color $[k]^n$ so that there is no monochromatic combinatorial line.

\end{enumerate}

\end{theorem}
Theorem~\ref{th:1} entails an alternative characterization of the Hales-Jewett theorem and its
density version:
\begin{theorem}[Hales-Jewett theorem]
\label{hj_cc}
For every fixed $k\ge 3$, one has $D(Part_{n,k}) = \omega(1)$.
\end{theorem}

\begin{theorem}[Density Hales-Jewett theorem]
\label{dhj_cc}
For every $k\ge 3$ there holds $$\lim_{n\to \infty} \ind(Part_{n,k})/k^n = 0.$$
\end{theorem}
Given the central role the Hales-Jewett theorem plays in Ramsey Theory, and the intricacy of its proof, 
it would be very nice to find a proof of the Hales-Jewett theorem in the framework of communication complexity.

The relation between $c_{n,k}$ and communication complexity also suggests a way to prove a 
lower bound on $c_{n,k}$: prove an efficient communication protocol for $Part_{n,k}$. We show indeed that $D(Part_{n,k}) \le O\left(\log n\right)^{1/\lceil \log_2 k \rceil}$,
and thus the lower bound follows. 
We prove the relationship between $c_{n,k}$ and $\ind(Part_{n,k})$ in Section~\ref{sec:1}, and the lower bound on $c_{n,k}$ is proved in 
Section~\ref{sec:2}. Lastly, Section~\ref{sec:3} contains a discussion on Fujimura sets, mentioned in ~\cite{polymath2010moser}, 
and their communication complexity analogues.

\section{A communication complexity version of Hales-Jewett}
\label{sec:1}

We start with the definition of a {\em star}: A star is a subset of $X_1 \times \cdots \times X_k$ of the form
\[\{(x'_1, x_2, \ldots, x_k), (x_1,x'_2, \ldots, x_k), \ldots, (x_1,x_2,\ldots, x'_k)\},\]
where $x_i \ne x'_i$ for each $i$. We refer to $(x_1, x_2, \ldots, x_k)$ as the star's {\em center}.
Cylinder intersections can be easily characterized in terms of stars.
\begin{lemma}[\cite{KN97}]
\label{cylinders_and_stars}
A subset $C \subseteq X_1 \times \cdots \times X_k$ is a cylinder intersection if and only if for every
star that is contained in $C$, its center also belongs to $C$.
\end{lemma}

The function $Part_{n,k}$ has the property that for every $S_1,\ldots,S_{k-1} \in 2^{[n]}$ there 
is at most one set $S \subset [n]$ such that $Part_{n,k}(S_1,\ldots,S_{k-1},S) = 1$. We call such a function
a {\em weak graph function}, as opposed to a {\em graph function} \cite{BDPW07} where there is always exactly one such $S$.

Graph functions have some particularly convenient properties, one of which is that  
$1$-monochromatic cylinder intersections are characterized simply by the existence of stars, as proved in \cite{hdp17}. The same proof also works for 
weak graph functions and gives: 
\begin{lemma}[\cite{hdp17}]
\label{lem_cyl_int_gf}
Let $f: X_1 \times \cdots \times X_k \to \{0,1\}$ be a weak graph function and $C \subseteq f^{-1}(1)$.
The set $C$ is a ($1$-monochromatic) cylinder intersection with respect to $f$ if and only if
it does not contain a star.
\end{lemma}

\begin{proof}
If $C$ does not contain a star, then $C$ is a cylinder intersection by Lemma~\ref{cylinders_and_stars}.
On the other hand, if $C$ contains a star whose center is $(x_1, x_2,\ldots,x_k)$, then
by definition of a weak graph function $f(x_1,x_2,\ldots,x_k)=0$. Thus,
$C$ does not contain the center of star, and therefore $C$ is not a cylinder
intersection (again using Lemma~\ref{cylinders_and_stars}).
\end{proof}

\begin{proof}[of Theorem~\ref{th:1}]
As in \cite{chattopadhyay2007languages}, define a bijection $\psi$ from $Part_{n,k}^{-1}(1)$ to $[k]^n$. A $k$-tuple $(S_1,\ldots,S_k)$
is mapped to $(j_1,\ldots,j_n) \in [k]^n$ where $j_i$ is the   
index of the set $S_{j_i}$ that contains $i$. Since $S_1,\ldots,S_k$ form a partition of $[n]$ this map is a bijection.

Now consider a $1$-monochromatic star $(S'_1,\ldots,S_k),\ldots,(S_1,\ldots,S'_k)$ with respect to $Part_{n,k}$.
Since this star is $1$-monochromatic, it implies that in each of the families $(S'_1,\ldots,S_k),\ldots,(S_1,\ldots,S'_k)$, all 
subsets are pairwise disjoint. As a result, because $k\ge 3$, we get that the subsets $(S_1,\ldots,S_k)$ are also pairwise disjoint. 
This determines $S'_j$ uniquely: $S'_j = S_j \cup ([n]\setminus (\cup_{i=1}^k S_i))$, for every $j=1,\ldots,k$. 
Therefore, if we consider $\psi(S'_1,\ldots,S_k),\ldots, \psi(S_1,\ldots,S'_k)$ and place them in a $k\times n$
table, then the columns of this table all belong to $\{ (x,x,\ldots,x) : x\in [k] \} \cup \{ (1,2,\ldots,k) \}$. 
The $i$-th column of this table is in $\{ (x,x,\ldots,x) : x\in [k] \}$ if $i \in S_j$ for some $j \in [k]$ and
otherwise the $i$-th column is equal to $(1,2,\ldots,k)$.
Thus the stars in $Part_{n,k}^{-1}(1)$ are mapped to combinatorial lines in $[k]^n$.

On the other hand, consider a combinatorial line in $[k]^n$ given by a $k\times n$ matrix $L$.
Let $L_1, \ldots, L_k$ be the rows of $L$, then it is not hard to check that similarly to the above,
$\psi^{-1}(L_1),\ldots,\psi^{-1}(L_k)$ form a $1$-monochromatic star with respect to $Part_{n,k}$.
The center of this star is $(S_1,\ldots,S_k)$ where $S_j$ contains all indices of columns that are equal
to $(j,j,\ldots,j)$.

Hence the stars in $Part_{n,k}^{-1}(1)$ are in one-to-one correspondence with combinatorial lines in $[k]^n$. 
It follows that $c_{n,k} = \ind(Part_{n,k})$, and that $\chr(Part_{n,k})$ is equal to the minimal number of
colors required to color $[k]^n$ so that there is no monochromatic combinatorial line.
\end{proof}

It is left to show the equivalence between Theorem~\ref{hj_cc} and the Hales-Jewett theorem,
and Theorem~\ref{dhj_cc} with its density version. 
The latter equivalence follows immediately from
part 1 of Theorem~\ref{th:1}. The equivalence of Theorem~\ref{hj_cc} to the Hales-Jewett theorem 
follows from part 2 of Theorem~\ref{th:1}, and the following theorem:
\begin{theorem}[\cite{hdp17}]
\label{D_chr_gf}
For every weak graph function $f:X_1 \times \cdots \times X_k \to \{0,1\}$, there holds
$$
\log \chr(f) \le D(f) \le \lceil \log \chr(f) \rceil + k.
$$
\end{theorem}
Theorem~\ref{D_chr_gf} was proved in \cite{hdp17} for graph functions. The same proof
with a minor change works for weak graph functions. For completeness we add the proof.
\begin{proof}
The lower bound is standard and holds for every function, 
thus it is only required to prove the upper bound.
Fix a $\chr(f)$-coloring of $f^{-1}(1)$ where every color class is star-free.
On input $x_1,x_2,\ldots,x_{k-1},y$, the last player 
first checks and announces whether there is a value 
$y'$ such that $f(x_1,x_2,\ldots,x_{k-1},y')=1$, using $1$ bit.
If $y'$ exists, the last player then computes and publishes the color $b$ of
$(x_1,x_2,\ldots,x_{k-1},y')$. 
If $y'$ does not exist, the protocol ends with value $0$.
Note that since $f$ is a weak graph function, if $y'$ exists then it is unique.

Then, for each $i=1,\ldots,k-1$, player $P_i$ checks whether there is a value $x_i'$ such that $f(x_1,x_2,\ldots,x_{i-1},x_i',x_{i+1},..,y)=1$
and $(x_1,x_2,\ldots,x_{i-1},x_i',x_{i+1},..,y)$ is colored $b$. He writes $1$ on the board if such
an $x_i'$ exists and writes $0$ otherwise. The protocol's value is $1$ if and only if all players wrote $1$ on the board.

The total number of bits communicated in this protocol is $\lceil \log \chr(f) \rceil + k$. We turn to prove that the
protocol is correct. When $f(x_1,x_2,\ldots,x_{k-1},y)=1$, the protocol clearly outputs $1$.
Now suppose that it outputs $1$, even though $f(x_1,x_2,\ldots,x_{k-1},y)=0$. This means
that there is a choice of $x_1',x_2',\ldots,y'$ for which
\[f(x_1',x_2,\ldots,x_{k-1},y)=f(x_1,x_2',\ldots,x_{k-1},y)=\ldots =f(x_1,x_2,\ldots,x_{k-1},y')=1,\]
and all points are in the same color set.
But then this color set in $f^{-1}(1)$ cannot
constitute a star-free set, a contradiction.
\end{proof}

\section{A lower bound on $c_{n,k}$}
\label{sec:2}

In this section we prove the following lower bound, similar to that of \cite{polymath2010moser}:
\begin{theorem}
\label{th:2}
For each $k \ge 3$, there is an absolute constant $C > 0$ such that
\[
c_{n,k} \ge Ck^n \frac{r_k(kn)}{kn \log kn}  = k^n \exp\left(-O(\log n)^{1/\lceil \log_2 k \rceil}\right).
\]
\end{theorem}
We first give an efficient protocol for $Part_{n,k}$, and then explain how it implies
Theorem~\ref{th:2}. 
\begin{lemma}
\label{lem:protocol_partnk}
For every fixed $k \ge 3$ it holds that
\[
D(Part_{n,k}) \le O\left(\log \frac{kn \log kn}{r_k(kn)} \right)  = O(\log n)^{1/\lceil \log_2 k \rceil}.
\]
\end{lemma}

\begin{proof} 
The protocol uses a known reduction to the Exactly-$n$ function, see e.g. \cite{beigel2006multiparty, chattopadhyay2007languages}.
Define $Exactly_{n,k}(x_1,\ldots,x_k)=1$ if and only if 
$\sum_{i=1}^k x_i=n$, where $(x_1,\ldots,x_k)$ are non-negative integers. 
The reduction is simple, given an instance $(S_1,\ldots,S_k)$ 
to be computed, the players do the following:
\begin{enumerate}

\item The $k$-th player checks whether $S_1,\ldots,S_{k-1}$ are pairwise disjoint. If they are not
pairwise disjoint then the protocol ends with rejection. 

\item The first player checks whether $S_2,\ldots,S_{k-1}$ are each disjoint from $S_k$. 
If this is not the case then the protocol ends with rejection.

\item The second player checks whether $S_1 \cap S_k  = \emptyset$ and rejects if not.

\item The players use a protocol for $Exactly_{n,k}$ to determine whether $\sum_{i=1}^k |S_i| = n$.
The protocol accepts if and only if equality holds, and the sum is exactly $n$. 

\end{enumerate}

The first three steps of the above protocol require three bits of communication, and the last part
uses a protocol for $Exactly_{n,k}$. Chandra, Furst and Lipton \cite{CFL83} gave a surprising protocol for $Exactly_{n,k}$
with at most $O(\log \frac{kn \log kn}{r_k(kn)})$ bits of communication.
It was later observed by Beigel, Gasarch and Glenn $\cite{beigel2006multiparty}$ that when plugging in the bounds on $r_k(n)$ given by the construction
of Rankin \cite{rankin1965sets} one gets $O(\log \frac{kn \log kn}{r_k(kn)})  = O(\log n)^{1/\lceil \log_2 k \rceil}$. 
\end{proof}

\begin{proof}[of Theorem~\ref{th:2}]
As mentioned before $\log \chr(f) \le D(f)$ holds for every function f, 
combined with Lemma~\ref{lem:protocol_partnk} this gives
\[
\chr(Part_{n,k}) \le \exp\left( D(Part_{n,k}) \right)\le O\left( \frac{kn \log kn}{r_k(kn)} \right) = \exp\left( O(\log n)^{1/\lceil \log_2 k \rceil}\right).
\]
Since $\ind(f) \ge |f^{-1}(1)|/\chr(f)$ holds also for every $f$ and $|Part_{n,k}^{-1}(1)| = k^n$, we get
\[
\ind(Part_{n,k}) \ge \frac{k^n}{\chr(Part_{n,k}) } \ge \Omega\left( k^n \frac{r_k(kn)}{kn \log kn} \right)  = k^n \exp\left(-O(\log n)^{1/\lceil \log_2 k \rceil}\right).
\]
The lower bound on $c_{n,k}$ now follows from part 1 of Theorem~\ref{th:1}.
\end{proof}

%\begin{remark}

\section{Fujimura sets}
\label{sec:3}

The following definitions are from \cite{polymath2010moser}.
Let $\Delta_{n,k}$ denote the set of $k$-tuples $(a_1,\ldots,a_k) \in \mathbb{N}^k$ such that
$\sum_{i=1}^k a_i = n$. Define a {\em simplex} to be a set of $k$ points in $\Delta_{n,k}$ of the form
$(a_1+r,a_2,\ldots,a_k),(a_1,a_2+r,\ldots,a_k),\ldots,(a_1,a_2,\ldots,a_k+r)$ for some $0<r \le n$. 
Define a Fujimura set to be a subset $B \subset \Delta_{n,k}$ that contains no simplices. 

Theorem~\ref{th:2} actually proves a lower bound on the maximal size of a {\em Fujimura set} 
in $\Delta_{n,k}$, similarly to the proof in \cite{polymath2010moser}. In fact
\begin{itemize}

\item $\Delta_{n,k} = (Exactly_{n,k})^{-1}(1)$.

\item A simplex in $\Delta_{n,k}$ is equivalent to a star.

\item $\ind(Exactly_{n,k})$ is equal to the maximal size of a Fujimura set in $\Delta_{n,k}$,
which is denoted by $c^{\mu}_{n,k}$ in \cite{polymath2010moser}. 

\end{itemize}
 
The proof of Theorem~\ref{th:2} gives essentially a lower bound for $\ind(Exactly_{n,k})$ via
an efficient protocol for Exactly-$n$, and the lower bound for $c_{n,k}$ is implied from the fact that
$D(Part_{n,k}) \le D(Exactly_{n,k}) + 3$. It is an interesting question whether this bound is tight, or 
is it the case that $D(Part_{n,k})$ can be significantly smaller than $D(Exactly_{n,k})$.
This is equivalent to asking whether lower bounds on the density Hales-Jewett number via bounds
on the maximal size of a Fujimura set can be tight, or close to tight.
In this respect it is interesting to note the following characterization of the communication complexity
of $Exactly_{n,k}$ in the language of $Part_{n,k}$.

Let $m \geq n$ be natural numbers, define the function $Part_{m,k,n}: (2^{[m]})^k \to \{0,1\}$ as follows, $Part_{m,k,n}(S_1,\ldots,S_k) = 1$ if and only if
$(S_1,\ldots,S_k)$ are pairwise disjoint and $\left| S_1 \cup S_2 \cup \ldots \cup S_k \right| = n$. 
Clearly, $Part_{n,k} = Part_{n,k,n}$ and as we observe in the next theorem $D(Exactly_{n,k})$
is also equivalent to the complexity of some function in this family. 

We call a map $g: [n]^k \to (2^{[m]})^k$ {\em sum preserving} if the following two properties 
hold for every $(a_1,\ldots,a_k) \in ([n])^k$: (i) $|g(a_1,\ldots,a_k)_i|=a_i$,
(ii) $g(a_1,\ldots,a_k)$ are pairwise disjoint whenever $\sum_{i=1}^k a_i = n$.
\begin{theorem}
Let $n$, $k$ and $m$ be natural numbers, if there exists a sum preserving map $g: [n]^k \to (2^{[m]})^k$. Then
$$D(Exactly_{n,k}) \le D(Part_{m,k,n}) \le D(Exactly_{n,k}) +3.$$
\end{theorem}

\begin{proof}
Similarly to the case $m=n$, given an instance $(S_1,\ldots,S_k)$ to be computed, the protocol is:
\begin{enumerate}

\item The players check whether $S_1,\ldots,S_{k}$ are pairwise disjoint, using three bits of communication. 
If they are not pairwise disjoint then the protocol ends with a rejection. 

\item The players use a protocol for $Exactly_{n,k}$ to determine whether $\sum_{i=1}^k |S_i| = n$.
The protocol accepts if and only if equality holds, and the sum is equal to $n$. 

\end{enumerate}
It follows that $D(Part_{m,k,n}) \le D(Exactly_{n,k})+3$. 
Note that $Part_{m_1,k,n} \le Part_{m_2,k,n}$ whenever $m_1 \le m_2$. 
Therefore the minimal value of $D(Part_{m,k,n})$ is achieved when $m=n$, i.e.
for $Part_{n,k}$.  

On the other direction, to get a protocol for $Exactly_{n,k}$ the players decide before hand
on a sum preserving map $g: [n]^k \to (2^{[m]})^k$.
Then, given an instance $(a_1,\ldots,a_k)$ to $Exactly_{n,k}$ the players solve the instance $g(a_1,\ldots,a_k)$
using an optimal protocol for $Part_{m,k,n}$. Since $g$ is sum preserving this reduction always gives the correct answer.
We conclude that $D(Exactly_{n,k}) \le D(Part_{m,k,n})$.
\end{proof}

Therefore, the question of separating $D(Part_{n,k})$ from $D(Exactly_{n,k})$ is actually a question
of separating $D(Part_{m,k,n})$ for different values of $m$. Notice that for $m = kn$ 
there already exists a sum preserving map $g: [n]^k \to (2^{[m]})^k$, simply take
$g(a_1,\ldots,a_k) = (\{1,\ldots,a_1\}, \{n+1,\ldots,n+a_2\}, \ldots, \{(k-1)n+1,\ldots,(k-1)n+a_k\})$.

Thus the ranges of $m$ we are interested in are $m\le kn$.
This can be improved to $m \le \lceil \frac{kn}{2} \rceil$, by pairing adjacent entries and considering the map
$g(a_1,\ldots,a_k) = (\{1,\ldots,a_1\}, \{n-1,\ldots,n-a_2\}, \{n+1,\ldots,n+a_3\}, \{2n-1,\ldots,2n-a_4\}, \ldots)$.
The sets in this case might intersect, but if they do it implies that the sum $\sum a_i$ is greater than
$n$ and the value $0$ is correct. For $k=3$ this gives $m \le 2n$ and the question is to separate
$D(Part_{n,k,n})$ from $D(Part_{2n,k,n})$.

\section{Discussion and open problems}

The relation between the NOF model of communication complexity and Ramsey theory and related areas
of mathematics was evident already in the initial paper of Chandra, Furst and Lipton \cite{CFL83}. Since then, 
the breadth and profoundness of this relation is better understood, see e.g. \cite{pudlak2003application, beigel2006multiparty, chattopadhyay2007languages, ada2015nof, hdp17}. 
This note offers another strong bridge, showing that the Hales-Jewett theorem, a pillar of Ramsey theory, and related questions, are naturally formulated
in this model of communication complexity. 

We already know that the NOF model is rich enough to formulate many interesting questions in the theory of 
computer science, e.g. proving lower bounds on the size of $ACC^0$ circuits \cite{HG91}. 
The new relations that we find give yet another proof to the richness and significance of this model,
not only in computer science. For these relations to bear fruit though, it is not enough to describe the problems 
in communication complexity language, we need to also develop
the tools to handle them in this setting. The main purpose of this note was to further motivate this study.
Some interesting open questions in the context of the problems described here, are:

\begin{enumerate}

\item Find a protocol for $Part_{n,k}$ that does not rely on the construction
of Behrend \cite{behrend1946sets} and Rankin \cite{rankin1965sets}.

\item Find more efficient protocols for $Part_{n,k}$, and thus improve the lower bound on the density Hales-Jewett number.
For $k > 3$ we believe that the protocol described here is not optimal.

\item Prove a lower bound for the communication complexity of $Part_{n,k}$ using 
communication complexity tools, e.g. via a reduction.
Currently the tools of communication complexity do not seem to even give $D_k(Part_{n,k}) \to \infty$.    

\item Determine the relation between $D_k(Part_{n,k})$ and $D_k(Exactly_{n,k})$. From the point of view
of communication complexity it makes sense to believe that these two are closely related, since determining whether
pairwise disjoint sets $S_1,\ldots,S_k$ form a partition of $[n]$ essentially amounts to verifying that 
$|S_1|+|S_2|+\cdots+|S_k| = n$. It is therefore reasonable to make the following conjecture (see 
Section~\ref{sec:3} for further discussion):

\begin{conjecture}
$D_k(Part_{n,k}) = \Theta(D_k(Exactly_{n,k}))$.
\end{conjecture} 

If the above conjecture is true, it in particular gives a strong proof of the Hales-Jewett theorem
as well as deep insight into the relation between the Hales-Jewett theorem and multidimensional
Szemer\'{e}di theorems.

\end{enumerate}

%\end{remark}

%\section{Discussion}

%The relation between Ramsey theory and communication complexity in the NOF model is evident
%from the first paper defining this model \cite{CFL83}. It became even clearer with time that this relation
%is not just possible applications of Ramsey theory in communication complexity, but that in fact the NOF
%model is a part of Ramsey theory. The relation seems to be revolve around the communication complexity
%of high dimensional permutations \cite{hdp17} and their derivatives. For example, the following quantities
%all have an analogue in this framework:

%\begin{enumerate}

%\item The maximal size of a corner free set in the grid $[n] \times [n]$.

%\item The maximal number of colors needed to color the grid $[n]\times [n]$ so that no
%corner is monochromatic.

%\item  

%\end{enumerate}  

%%%%%%%%%%%%%%%%%%%%%%%%%%%%%%%%%%%%%%%%%%%%%%%%%%%%%%%%%%%%%%%%%%%%%%%%%%%%%%%%%%%%%%%%%%%%%%%%%%%%%%%%%%%%%%%%%
% Bibliography
%%%%%%%%%%%%%%%%%%%%%%%%%%%%%%%%%%%%%%%%%%%%%%%%%%%%%%%%%%%%%%%%%%%%%%%%%%%%%%%%%%%%%%%%%%%%%%%%%%%%%%%%%%%%%%%%%

%\bibliographystyle{alpha}
\bibliographystyle{plain}
\bibliography{../complexity}

\begin{thebibliography}{10}

\bibitem{ada2015nof}
A.~Ada, A.~Chattopadhyay, O.~Fawzi, and P.~Nguyen.
\newblock The nof multiparty communication complexity of composed functions.
\newblock {\em computational complexity}, 24(3):645--694, 2015.

\bibitem{austin2011deducing}
T.~Austin.
\newblock Deducing the density {Hales--Jewett} theorem from an infinitary
  removal lemma.
\newblock {\em Journal of Theoretical Probability}, 24(3):615--633, 2011.

\bibitem{BDPW07}
P.~Beame, M.~David, T.~Pitassi, and P.~Woelfel.
\newblock Separating deterministic from randomized nof multiparty communication
  complexity.
\newblock In {\em Proceedings of the 34th International Colloquium On Automata,
  Languages and Programming}, Lecture Notes in Computer Science.
  Springer-Verlag, 2007.

\bibitem{behrend1946sets}
F.~A. Behrend.
\newblock On sets of integers which contain no three terms in arithmetical
  progression.
\newblock {\em Proceedings of the National Academy of Sciences},
  32(12):331--332, 1946.

\bibitem{beigel2006multiparty}
R.~Beigel, W.~Gasarch, and J.~Glenn.
\newblock The multiparty communication complexity of {Exact-T}: Improved bounds
  and new problems.
\newblock In {\em International Symposium on Mathematical Foundations of
  Computer Science}, pages 146--156. Springer, 2006.

\bibitem{CFL83}
A.~Chandra, M.~Furst, and R.~Lipton.
\newblock Multi-party protocols.
\newblock In {\em Proceedings of the 15th ACM Symposium on the Theory of
  Computing}, pages 94--99. ACM, 1983.

\bibitem{chattopadhyay2007languages}
A.~Chattopadhyay, A.~Krebs, M.~Kouck{\`y}, M.~Szegedy, P.~Tesson, and
  D.~Th{\'e}rien.
\newblock Languages with bounded multiparty communication complexity.
\newblock In {\em Annual Symposium on Theoretical Aspects of Computer Science},
  pages 500--511. Springer, 2007.

\bibitem{dodos2013simple}
P.~Dodos, V.~Kanellopoulos, and K.~Tyros.
\newblock A simple proof of the density {Hales--Jewett} theorem.
\newblock {\em International Mathematics Research Notices},
  2014(12):3340--3352, 2013.

\bibitem{elkin2010improved}
M.~Elkin.
\newblock An improved construction of progression-free sets.
\newblock In {\em Proceedings of the twenty-first annual ACM-SIAM symposium on
  Discrete Algorithms}, pages 886--905. Society for Industrial and Applied
  Mathematics, 2010.

\bibitem{furstenberg1978ergodic}
H.~Furstenberg and Y.~Katznelson.
\newblock An ergodic {Szemer{\'e}di} theorem for commuting transformations.
\newblock {\em Journal d'Analyse Mathematique}, 34(1):275--291, 1978.

\bibitem{furstenberg1989density}
H.~Furstenberg and Y.~Katznelson.
\newblock A density version of the {Hales-Jewett} theorem for k=3.
\newblock {\em Discrete Mathematics}, 75(1-3):227--241, 1989.

\bibitem{furstenberg1991density}
H.~Furstenberg and Y.~Katznelson.
\newblock A density version of the {Hales-Jewett} theorem.
\newblock {\em Journal d'Analyse Mathematique}, 57(1):64--119, 1991.

\bibitem{green2010note}
B.~Green and J.~Wolf.
\newblock A note on {Elkin's} improvement of {Behrend's} construction.
\newblock In {\em Additive number theory}, pages 141--144. Springer, 2010.

\bibitem{HG91}
J.~H{\aa}stad and M.~Goldmann.
\newblock On the power of small-depth threshold circuits.
\newblock {\em Computational Complexity}, 1:113--129, 1991.

\bibitem{KN97}
E.~Kushilevitz and N.~Nisan.
\newblock {\em Communication Complexity}.
\newblock Cambridge University Press, 1997.

\bibitem{hdp17}
N.~Linial and A.~Shraibman.
\newblock On the communication complexity of high-dimensional permutations.
\newblock {\em arXiv preprint arXiv:1706.02207}, 2017.

\bibitem{o2011sets}
K.~O'Bryant.
\newblock Sets of integers that do not contain long arithmetic progressions.
\newblock {\em The electronic journal of combinatorics}, 18(1):59, 2011.

\bibitem{polymath2009new}
D.H.J. Polymath.
\newblock A new proof of the density {Hales-Jewett} theorem.
\newblock {\em arXiv preprint arXiv:0910.3926}, 2009.

\bibitem{polymath2010moser}
D.H.J. Polymath.
\newblock Density {Hales-Jewett} and {Moser} numbers.
\newblock {\em arXiv preprint arXiv:1002.0374v2}, 2010.

\bibitem{pudlak2003application}
P.~Pudl{\'a}k.
\newblock An application of hindman's theorem to a problem on communication
  complexity.
\newblock {\em Combinatorics, Probability and Computing}, 12(5+6):661--670,
  2003.

\bibitem{rankin1965sets}
R.~A. Rankin.
\newblock Sets of integers containing not more than a given number of terms in
  arithmetical progression.
\newblock {\em Proceedings of the Royal Society of Edinburgh Section A:
  Mathematics}, 65(4):332--344, 1961.

\bibitem{szemeredi1975sets}
E.~Szemer{\'e}di.
\newblock On sets of integers containing no k elements in arithmetic
  progression.
\newblock {\em Acta Arith}, 27(199-245):2, 1975.

\bibitem{van1927beweis}
B.~L. van~der Waerden.
\newblock Beweis einer baudetschen vermutung.
\newblock {\em Nieuw Arch. Wiskunde}, 15:212--216, 1927.

\end{thebibliography}

%%%%%%%%%%%%%%%%%%%%%%%%%%%%%%%%%%%%%%%%%%%%%%%%%%%%%%%%%%%%%%%%%%%%%%%%%%%%%%%%%%%%%%%%%%%%%%%%%%%%%%%%%%%%%%%%%

\end{document}